\documentclass[nofootinbib,a4paper,reprint,12pt,aps,tightenlines,onecolumn,longbibliography]{revtex4-1}
\usepackage[utf8]{inputenc}
\usepackage{amsmath}
\usepackage{amssymb}
\usepackage{cancel}
\usepackage{stmaryrd}
\usepackage{amsfonts}
\usepackage{bm}
\usepackage{color}
\usepackage{graphicx}
\DeclareSymbolFont{operators}{OT1}{cmr}{m}{n}
\DeclareSymbolFont{letters}{OML}{cmm}{m}{it}
\DeclareSymbolFont{symbols}{OMS}{cmsy}{m}{n}
\DeclareSymbolFont{largesymbols}{OMX}{cmex}{m}{n}
\usepackage[hyperindex,breaklinks]{hyperref}
\hypersetup{
    pdfborder={0 0 0},
    allbordercolors={0 0 0},         
    unicode=false,          
    pdftoolbar=true,        
    pdfmenubar=true,        
    pdffitwindow=false,     
    pdfstartview={FitH},    
    pdfcreator={Eero Hirvijoki},
    pdfproducer={pdflatex},
    pdfnewwindow=true,
    colorlinks=true,
    linkcolor=blue,
    citecolor=blue,
    filecolor=blue,
    urlcolor=blue,  
}
\usepackage[mathscr]{eucal}
\usepackage{array}
\usepackage{amsthm}
\usepackage{amsmath}
\usepackage{tikz}
\usepackage[colorinlistoftodos,prependcaption]{todonotes}
\usetikzlibrary{matrix}

\newtheorem{proposition}{Proposition}[section]
\newtheorem{definition}{Definition}[section]
\newtheorem{corollary}{Corollary}[section]
\theoremstyle{remark}

\newcommand{\lie}{\ensuremath{\mathsterling}}
\newcommand{\extd}{\ensuremath{\mathbf{d}}}
\newcommand{\iprod}{\ensuremath{\boldsymbol{\iota}}}

\newcommand{\R}{\mathbb{R}}

\newcommand{\rev}[1]{\textcolor{black}{#1}}

\makeatother

\begin{document}

\title{Energy and momentum conservation in the Euler-Poincar\'e formulation of local Vlasov-Maxwell-type systems}

\author{Eero Hirvijoki}
\affiliation{Department of Applied Physics, Aalto University, P.O. Box 11100, 00076 AALTO, Finland}
\email{eero.hirvijoki@gmail.com}

\author{Joshua W. Burby}
\affiliation{Los Alamos National Laboratory, Los Alamos, New Mexico 87545,
USA}

\author{David Pfefferl\'e}
\affiliation{The University of Western Australia, 35 Stirling Highway, Crawley WA 6009, Australia}

\author{Alain J. Brizard}
\affiliation{Department of Physics, Saint Michael's College, Colchester, VT 05439, USA}
\date{\today}

\begin{abstract}
    The action principle by Low [Proc. R. Soc. Lond. A 248, 282--287] for the classic Vlasov-Maxwell system contains a mix of Eulerian and Lagrangian variables. This renders the Noether analysis of reparametrization symmetries inconvenient, especially since the well-known energy- and momentum-conservation laws for the system are expressed in terms of Eulerian variables only. While an Euler-Poincar\'e formulation of Vlasov-Maxwell-type systems, effectively starting with Low's action and using constrained variations for the Eulerian description of particle motion, has been known for a while [J. Math. Phys., 39, 6, pp. 3138-3157], it is hard to come by a documented derivation of the related energy- and momentum-conservation laws in the spirit of the Euler-Poincar\'e machinery. To our knowledge only one such derivation exists in the literature so far, dealing with the so-called guiding-center Vlasov-Darwin system [Phys. Plasmas 25, 102506]. The present exposition discusses a generic class of local Vlasov-Maxwell-type systems, with a conscious choice of adopting the language of differential geometry to exploit the Euler-Poincar\'e framework to its full extent. After reviewing the transition from a Lagrangian picture to an Eulerian one, we demonstrate how symmetries generated by isometries in space lead to conservation laws for linear- and angular-momentum density and how symmetry by time translation produces a conservation law for energy density. We also discuss what happens if no symmetries exist. Finally, two explicit examples will be given -- the classic Vlasov-Maxwell and the drift-kinetic Vlasov-Maxwell -- and the results expressed in the language of regular vector calculus for familiarity.
\end{abstract}

\maketitle

\section{Introduction}
Recall that the Vlasov-Maxwell system couples an advection equation for particle phase-space number density $F(\bm{x},\bm{v},t)d^3\bm{x}d^3\bm{v}$ to Maxwell's equations for the electromagnetic fields in a self-consistent manner: the current and charge densities in Maxwell's equations are computed as velocity-space moments of the particle distribution function, according to $\varrho=e\int_{\bm{v}}F d^3\bm{v}$ and $\bm{j}=e\int_{\bm{v}}\bm{v}F d^3\bm{v}$ \rev{(the summation over species is implicitly assumed and henceforth omitted)}, and the Lorentz force responsible for the particle trajectories depends on the fields $\bm{E}$ and $\bm{B}$. The set of equations, governing the dynamics and constraints of the system, becomes
\begin{subequations}
\label{eq:VM}
\begin{align}
    \partial_tF+\nabla\cdot(\bm{v}F)+\partial_{\bm{v}}\cdot\left((e/m)(\bm{E}+\bm{v}\times\bm{B})F\right)&=0,\\
    \varepsilon_0\partial_t\bm{E}+\bm{j}-\mu_0^{-1}\nabla\times\bm{B}&=0,\\
    \partial_t\bm{B}+\nabla\times\bm{E}&=0,\\
    \varepsilon_0\nabla\cdot\bm{E}-\varrho&=0,\\
    \nabla\cdot\bm{B}&=0.
\end{align}
\end{subequations}

Conservation laws for this system are straightforward to identify directly from the equations of motion, with a bit of intuition. Multiplying the advection equation for $F$ with $m\bm{v}$ and $\tfrac{1}{2}m|\bm{v}|^2$, and integrating over the velocity space, one finds
\begin{align}
    &\partial_t\int m\bm{v}Fd^3\bm{v}+\nabla\cdot\int m\bm{v}\bm{v}Fd^3\bm{v}=\varrho\bm{E}+\bm{j}\times\bm{B},\\
    &\partial_t\int \frac{1}{2}m|\bm{v}|^2Fd^3\bm{v}+\nabla\cdot\int \frac{1}{2}m|\bm{v}|^2\bm{v}Fd^3\bm{v}=\bm{j}\cdot\bm{E}.
\end{align}
On the other hand, an educated guess and Maxwell's equations demonstrate that
\begin{align}
    &\bm{E}\cdot\bm{j}+\frac{1}{2}\partial_t\big(\varepsilon_0|\bm{E}|^2+\mu_0^{-1}|\bm{B}|^2\big)=-\nabla\cdot(\mu_0^{-1}\bm{E}\times\bm{B}),\\
    &\varrho\bm{E}+\bm{j}\times\bm{B}+\partial_t(\varepsilon_0\bm{E}\times\bm{B})=-\nabla\cdot\Big(\frac{1}{2}(\varepsilon_0|\bm{E}|^2+\mu_0^{-1}|\bm{B}|^2)\mathbf{1}-\mu_0^{-1}\bm{B}\bm{B}-\varepsilon_0\bm{E}\bm{E}\Big),
\end{align}
where $\mathbf{1}$ is the identity dyad. When the expressions above are combined, local conservation laws for linear momentum density and energy density are obtained
\begin{align}
    \partial_t\left(\int m\bm{v}Fd\bm{v}+\varepsilon_0\mu_0\bm{S}\right)+\nabla\cdot\left(\int m\bm{v}\bm{v}Fd\bm{v}-\bm{\mathcal{E}}\right)&=0,\\
    \partial_t\left(\int\frac{1}{2}m|\bm{v}|^2Fd\bm{v}-\text{Tr}(\bm{\mathcal{E}})\right)+\nabla\cdot\left(\int\frac{1}{2}m|\bm{v}|^2\bm{v}Fd\bm{v}+\bm{S}\right)&=0,
\end{align}
where $\text{Tr}(\cdot)$ is the trace and the Maxwell stress tensor $\bm{\mathcal{E}}$ and the Poynting vector $\bm{S}$ are
\begin{align}
    \bm{\mathcal{E}}&=-\frac{1}{2}\left(\varepsilon_0|\bm{E}|^2+\mu_0^{-1}|\bm{B}|^2\right)\mathbf{1}+\varepsilon_0\bm{E}\bm{E}+\mu_0^{-1}\bm{B}\bm{B},\\
    \bm{S}&=\mu_{0}^{-1}\bm{E}\times\bm{B}.
\end{align}
Conservation of angular momentum with respect to a given axis follows immediately from the symmetry of $\bm{v}\bm{v}$ and $\mathcal{E}$\footnote{\rev{In right-handed cylindrical coordinates $(r,\varphi,z)$, let $\nabla z$ be axis of rotation and $\varphi$ the angle of rotation. Then for any symmetric second order tensor $\mathbf{T}$ we have $(\nabla\cdot\mathbf{T})\cdot\partial_{\varphi}\bm{x}=\nabla\cdot(\mathbf{T}\cdot\partial_{\varphi}\bm{x})-\mathbf{T}:\nabla\partial_{\varphi}\bm{x}=\nabla\cdot(\mathbf{T}\cdot\partial_{\varphi}\bm{x})$ because $\nabla\partial_{\varphi}\bm{x}=r(\nabla r\nabla\varphi-\nabla\varphi\nabla r)$ is antisymmetric, and angular momentum conservation follows.}}.

While the results above were easy to come by, it is preferable to obtain them directly from a variational principle using Noether's theorem. This systematic strategy is especially useful when dealing with alternate Vlasov-Maxwell-type systems where the particle motion couples to electromagnetic fields in a far more complicated way, blurring the intuition for making an educated guess. At least four such Vlasov-Maxwell systems exist and can be used in numerical modeling of plasmas in various branches of science. These are the guiding-center \cite{Brizard-Tronci:2016}, the drift-kinetic \cite{Burby-initial-value-problem:2016JPlPh,Burby-Brizard:2019PhLA}, the gyrokinetic \cite{Burby-et-al:2015PhLA,Burby-Brizard:2019PhLA}, and the spin-Vlasov-Maxwell system \cite{Burby-finite-dimensional:2017PhPl}. They all have a structure similar to equations \eqref{eq:VM}. 

Over the years, several papers discussing action principles for the Vlasov-Maxwell system or related ones\footnote{By systems related to Vlasov-Maxwell models, we mean 1) genuine Vlasov-Maxwell models that form an infinite-dimensional initial-value problem for the dynamical variables, and 2) the so-called Vlasov-Poisson-Amp\`ere models which provide an initial value problem for the distribution function only and constraint equations for the electromagnetic potentials.} have been presented \cite{Low:1958,galloway_kim:1971,Pfirsch:1984ZNatA,Pfirsch_Morrison:1985PhRvA,Elvsen_Larsson:1993PhyS,Larsson:1993JPlPh,Fla:1994,Cendra_et_al:1998,Sugama:2000PhPl,Brizard:2000PRL,Brizard:2000PhPl,Sugama-et-al:2013PhPl,Sugama-et-al:2015PhPl,Burby-et-al:2015PhLA,Burby-initial-value-problem:2016JPlPh,Brizard-Tronci:2016,Burby-Sengupta:2018PhPl,Sugama_et_al_2018PhPl,Burby-Brizard:2019PhLA} and many of them~\cite{Pfirsch:1984ZNatA,Pfirsch_Morrison:1985PhRvA,Sugama:2000PhPl,Brizard:2000PRL,Brizard:2000PhPl,Sugama-et-al:2013PhPl,Sugama-et-al:2015PhPl,Brizard-Tronci:2016,Sugama_et_al_2018PhPl} discuss the local energy and momentum conservation laws. Nevertheless, to our knowledge the only documented work dealing with the conservation laws that has been carried out in the spirit of Euler-Poincar\'e formalism is the recent paper by Sugama et al. focusing on the guiding-center Vlasov-Darwin model \cite{Sugama_et_al_2018PhPl}. To continue filling the information vacuum, the present paper discusses a generic class of local Vlasov-Maxwell-type systems, with a conscious choice of adopting the language of differential geometry to exploit the Euler-Poincar\'e framework to its full extent. The reason we focus on genuine Vlasov-Maxwell type systems is their invariance under electromagnetic gauge transformations. This property together with compatible discretization schemes has opened new avenues in numerical plasma simulations (see, e.g., \cite{GEMPIC-review:2018PlST} and references therein). 

We will start from a modification of Low's action principle for Vlasov-Maxwell-type systems and, after reviewing the transition from a Lagrangian picture to an Eulerian one, we demonstrate how space-time-isometry symmetries in the action functional lead to conservation laws for linear- and  angular-momentum  density and for energy density. We will also discuss what happens if no such symmetry with respect to an isometry exists. Once this process is finished, we hope to have demonstrated how powerful the Euler-Poincar\'e framework can be in the context of kinetic plasma theories and how elegantly its geometric exposition suits the study of space-time symmetries.

Finally, two explicit examples will be given---the classic Vlasov-Maxwell \rev{in an axially symmetric background magnetic field} and the drift-kinetic Vlasov-Maxwell that is obtainable as the long-wave-length limit of the non-local gyrokinetic theory \cite{Burby-Brizard:2019PhLA}---and the results expressed in the language of regular vector calculus for familiarity. The reason for focusing on these two systems is because of their robustness, fidelity and efficiency in kinetic simulations of magnetized plasmas. Combining a full Larmor model of ions and a drift-center description of electrons avoids many complications due to the non-local nature of gyrokinetic theories and, at the same time, eliminates the electron-cyclotron-frequency time scale. This combination has been made possible thanks to recently developed electromagnetically gauge-invariant gyrokinetic theory~\cite{Burby-Brizard:2019PhLA}.

The derivation of the guiding-center Vlasov-Maxwell model is a straightforward application of our general procedure and is hence omitted.

\section{Euler-Poincar\'e formulation of the action principle}
We start with a slightly modified version of Low's action principle \cite{Low:1958}. The purpose of the modification is to introduce the capability to handle a wider class of Vlasov-Maxwell-type systems, such as the classic full-particle and the drift-kinetic Vlasov-Maxwell systems. In what follows, all dynamical variables (time-dependent) are denoted by the subscript $t$ to clearly separate them from parameters and/or integration labels. 

\subsection{Action in a mixed-variable representation}
In the action principle, the single-particle phase-space Lagrangian is first multiplied by the phase-space density of fixed-value particle labels, then integrated over all of the particle's phase-space and a given time interval, and finally combined with the standard electromagnetic action to account for electromagnetic interactions in a self-consistent way. In such systems, the electromagnetic fields are treated as Eulerian variables and the role of the single-particle action is to carry (advect) the fixed-value phase-space-density labels along the phase-space flow of individual particles. 

The action is a functional of the particle's phase-space trajectory $\bm{z}_t$, the vector potential $\bm{A}_t$, the scalar potential $\phi_t$, and depends parametrically on the fixed-value density $F$. Written in a general form, we have
\begin{align}\label{eq:Low-action}
    S_F[\bm{z}_t,\bm{A}_t,\phi_t]=&\int\limits_{t_1}^{t_2}\int\limits_{\mathbb{R}^3}\int\limits_{\mathbb{R}^3}\Big(\vartheta_\alpha(\bm{z}_t(\bm{z}))\partial_tz_t^\alpha(\bm{z})-K(\bm{z}_t(\bm{z}),\bm{E}_t(\bm{x}_t(\bm{z})),\bm{B}_t(\bm{x}_t(\bm{z})))\Big)F(\bm{z})d^6\bm{z}dt\nonumber
    \\&+\int\limits_{t_1}^{t_2}\int\limits_{\mathbb{R}^3}\int\limits_{\mathbb{R}^3}\Big(e\bm{A}_{t}(\bm{x}_t(\bm{z}))\cdot\partial_t\bm{x}_t(\bm{z})-e\phi_t(\bm{x}_t(\bm{z}))\Big)F(\bm{z})d^6\bm{z}dt\nonumber
    \\&+\int\limits_{t_1}^{t_2}\int\limits_{\mathbb{R}^3}\frac{1}{2}(\varepsilon_0|\bm{E}_t(\bm{x})|^2-\mu_0^{-1}|\bm{B}_{\text{ext}}(\bm{x})+\bm{B}_t(\bm{x})|^2)d^3\bm{x}dt.
\end{align}
Here $\bm{z}=\{z^{\alpha}\}_{\alpha=1}^6=(\bm{x},\bm{v})=(\{x^i\}_{i=1}^3,\{v^i\}_{i=1}^3)$ are integration labels in the phase-space, $\bm{z}_t=\{z_t^{\alpha}\}_{\alpha=1}^6=(\bm{x}_t,\bm{v}_t)=(\{x_t^i\}_{i=1}^3,\{v_t^i\}_{i=1}^3)$ are the time-dependent phase-space coordinates of a single particle with $\partial_t\bm{z}_t=(\partial_t\bm{x}_t,\partial_t\bm{v}_t)$ as time derivatives (Eulerian phase-space velocities), and $\bm{z}_t(\bm{z})=(\bm{x}_t(\bm{z}),\bm{v}_t(\bm{z}))$ refers to the coordinates the particle would reach in time $t$ when starting from an initial point $\bm{z}$. The notation $F(\bm{z})d^6\bm{z}=F(\bm{z})d^3\bm{x}d^3\bm{v}$ denotes the phase-space density of the fixed-value labels, with the bare volume elements being $d^3\bm{x}=dx^1dx^2dx^3$ and $d^3\bm{v}=dv^1dv^2dv^3$. The dynamical electric and magnetic fields are derived from the potentials via the standard relations $\bm{E}_t=-\partial_t{\bm{A}}_t-\nabla\phi_t$. The external magnetic field emanates from an external, static vector potential $\bm{B}_{\text{ext}}=\nabla\times\bm{A}_{\text{ext}}$ with no external electric field present. The dot $\cdot$ refers to the Euclidean inner product of vectors in~$\R^3$.

The original phase-space formulation of the classic Vlasov-Maxwell system would be recovered by setting $\bm{A}_{\text{ext}}=0$, and choosing the functions $\vartheta_{\alpha}$ and $K$ so that $\vartheta_{\alpha}(\bm{z}_t(\bm{z}))\partial_tz^{\alpha}_t(\bm{z})=m\bm{v}_t(\bm{z})\cdot\partial_t\bm{x}_t(\bm{z})$ and $K(\bm{z}_t(\bm{z}))=\tfrac{1}{2}m|\bm{v}_t(\bm{z})|^2$. One can then interpret the first row of \eqref{eq:Low-action} to represent the free-particle action, the second row the coupling term to the electromagnetic fields, and the last row the electromagnetic action in a vacuum. Our modifications effectively affect only the "free-particle" action, where we allow the kinetic energy to depend locally on the dynamic electric and magnetic fields and the functions $\vartheta_{\alpha}$ to possibly depend on the whole phase-space, in anticipation of how the velocity vector $\bm{v}_t$ in guiding-center dynamics is defined with respect to a fairly unique choice of the coordinates $(v^1,v^2,v^3)$. 

One could apply Hamilton's principle directly to \eqref{eq:Low-action} and derive the related Euler-Lagrange conditions for the trio $(\bm{z}_t,\bm{A}_t,\phi_t)$. This approach will not yield the Vlasov equation directly though, as the source terms appearing in the Maxwell's equations involves integration of the fixed-value density $F$ over the initial phase-space coordinates. In this picture, a Noether-type analysis of symmetries rapidly becomes intricate via the space-time reparametrization of trajectories and fields. It is thus helpful to convert the action above and apply Hamilton's principle and Noether's theorem directly in terms of Eulerian variables.  

\subsection{Conversion to Eulerian variables}
The process is initiated by identifying different coordinate functions that appear in \eqref{eq:Low-action} with their differential-geometric counterparts. We list these elements and their interpretation as follows:
\begin{enumerate}
    \item The phase-space integration domain, namely the open set $\{(\bm{x},\bm{v})|\bm{x}\in\mathbb{R}^3,\bm{v}\in\mathbb{R}^3\}$, is identified as the tangent bundle $TQ=\bigcup\{(x,v_x)|x\in Q,v_x\in T_xQ\}$ of the manifold $Q=\mathbb{R}^3$. Unbolded symbols will denote representative elements, e.g. a point $x\in Q$, a tangent vector $v_x\in T_xQ$ at point $x\in Q$, and a generic point $z\in TQ$ on the tangent bundle.

    \item The time-dependent functions $z_t^\alpha$, representing a single-particle phase-space trajectory in $\R^6$, are interpreted as the local coordinates of a time-dependent diffeomorphisms $g_t\in\text{Diff}(TQ)$, namely a family of smooth maps $g_t:TQ\to TQ$ with smooth inverse such that $g_t^{-1}(g_t(z))=g_t(g_t^{-1}(z))=z$ for all $t\in(t_1,t_2)\subset\mathbb{R}$ and for all $z\in TQ$. For a fixed point $z \in TQ$, the time derivative of the diffeomorphism generates a tangent vector $\partial_tg_t(z) \in T_{g_t(z)}TQ$. We then construct the \emph{Eulerian velocity field} $\xi_t=\partial_t g_t\circ g_t^{-1}\in\mathfrak{X}(TQ)$ such that $\partial_t g_t(z) = {\xi_t}(g_t(z)) = \xi_t(z_t)$.  If this vector field has a coordinate representation $\xi_t=\xi_t^{\alpha}(\bm{z})\partial/\partial z^{\alpha}$, then $\partial_tz^{\alpha}_t(\bm{z})=\xi_t^{\alpha}(\bm{z}_t(\bm{z}))$.
    
    \item The scalar potential $\phi_t$ and the vector potential $\bm{A}_t(\bm{x})=A_{t,i}(\bm{x})\bm{e}^i(\bm{x})$ (written in so-called covariant components) are identified respectively as a time-dependent zero-form $\phi_t\in\Omega^0(Q)$ and as a time-dependent one-form $A_t\in\Omega^1(Q)$, locally expressed as ${A_t}_x=A_{t,i}(\bm{x})\extd x^i\in T^{\ast}_xQ$. The related time-dependent electric-field one-form $E_t=-\partial_tA_t-\extd \phi_t\in\Omega^1(Q)$ and time-dependent magnetic-field two-form $B_t=\extd A_t\in\Omega^2(Q)$ are also introduced. 
    
    \item The canonical projection map $\pi:TQ\rightarrow Q$, $(x,v)\mapsto \pi(x,v) = x$ is used to promote the electromagnetic potentials to differential forms on the tangent bundle, namely $\pi^*\phi_t = \phi_t \circ \pi \in \Omega^0(TQ)$ and $\pi^* A_t \in \Omega^1(TQ)$. This permits the identification $\phi_t(\bm{x}_t(z)) = \pi^*\phi_t(g_t(z)) = g_t^*\pi^*\phi_t(z)$ as a function on $TQ$. Using the tangent map of the canonical projection $T\pi:TTQ\to TQ$ such that $T\pi_{g_t(z)}(\partial_t g_t(z)) = \partial_t x_t(z) \in T_{\pi(g_t(z))}Q$, we identify $\bm{A}_t(\bm{x}_t(z))\cdot \partial_t\bm{x}_t(z) = {A_t}_{\pi(g_t(z))} ( T\pi_{g_t(z)}(\partial_t g_t(z))) = {\pi^*A_t}_{g_t(z)}(\partial_t g_t(z)) = \iprod_{\xi_t}\pi^*A_t(g_t(z))=g_t^*(\iprod_{\xi_t}\pi^*A_t)(z) = g_t^*(\iprod_{\partial_t g_t\circ g_t^{-1}}\pi^*A_t)(z)$ as a function on $TQ$.

    \item The fixed phase-space volume form $f\in\Omega^6(TQ)$ is introduced and, in local coordinates, has the expression $f_z=F(\bm{z})\extd x^1\wedge\extd x^2\wedge\extd x^3\wedge \extd v^1\wedge \extd v^2\wedge \extd v^3$.
    
    \item  We denote the function $K_t \in \Omega^0(TQ)$, which depends parametrically on the electromagnetic forms through the rule $K_t(z) = K(z,\pi^* E_t(z),\pi^*B_t(z))$. 
    We then identify the term $K(\bm{z}_t(z),\bm{E}_t(\bm{x}_t(z)),\bm{B}_t(\bm{x}_t(z))) = K_t(g_t(z)) = g_t^*K_t(z)$ as a function on $TQ$. 
    
    \item The functions $\vartheta_{\alpha}$ are analoguously viewed as the components of a phase-space one-form $\vartheta\in\Omega^1(TQ)$ expressed in local coordinates as $\vartheta_z=\vartheta_{\alpha}(\bm{z})\extd z^{\alpha}\in T_z^{\ast}TQ$. We view $\vartheta_\alpha(\bm{z}_t(z))\partial_t z^\alpha_t(z) = \vartheta_{g_t(z)}(\xi_t(g_t(z))) =(\iprod_{\xi_t}\vartheta)(g_t(z)) = g_t^*(\iprod_{\xi_t}\vartheta)(z) = g_t^*(\iprod_{\partial_t g_t\circ g_t^{-1}} \vartheta)(z)$ as a function on $TQ$.
    
\end{enumerate}
The electromagnetic part of the action, the third line in \eqref{eq:Low-action}, when written in geometric terms, becomes
\begin{align}
    S_{EM}[A_t,\phi_t]=\int\limits_{t_1}^{t_2}\int\limits_{Q}\frac{1}{2}\big(\varepsilon_0 E_t\wedge\star E_t-\mu_0^{-1}(B_{\text{ext}}+B_t)\wedge\star(B_{\text{ext}}+B_t)\big)dt,
\end{align}
where $\star:\Omega^k(Q)\to \Omega^{n-k}(Q)$ is the Hodge star operator induced by the Riemannian metric on $Q$.

Conversion of the first and second line of \eqref{eq:Low-action} proceeds by substituting the definitions from the list above and using the change of coordinates formula on the manifold so that the entire action can be written as
\begin{align}\label{eq:Low-action-geometric}
S_F[\bm{z}_t,\bm{A}_t,\phi_t]=&
\int\limits_{t_1}^{t_2}\int\limits_{TQ}
g_t^*\big(\iprod_{\partial_tg_t\circ g_t^{-1}}(\vartheta+e\pi^*A_t) -(K_t+e\pi^*\phi_t)\big)f dt+ S_{EM}[A_t,\phi_t]\nonumber 
    \\
    =&\int\limits_{t_1}^{t_2}\int\limits_{\rev{\text{Im}_{g_t}(TQ)}}
\big(\iprod_{\partial_tg_t\circ g_t^{-1}}(\vartheta+e\pi^*A_t )-(K_t+e\pi^*\phi_t)\big)g_{t\ast}f dt + S_{EM}[A_t,\phi_t]\nonumber 
    \\
    =&\int\limits_{t_1}^{t_2}\int\limits_{TQ}
\big(\iprod_{\partial_tg_t\circ g_t^{-1}}(\vartheta+e\pi^*A_t )-(K_t+e\pi^*\phi_t)\big)g_{t\ast}f dt+ S_{EM}[A_t,\phi_t]\nonumber 
    \\
    =&S_f[g_t,A_t,\phi_t]
\end{align}
The last step follows from the fact that diffeomorphisms are one-to-one maps, meaning that $\text{Im}_{g_t}(TQ)=TQ$. Here the subscript $f$ in $S_f[g_t,A_t,\phi_t]$ stresses the parametric dependency on the fixed volume form $f$, rather than on the scalar $F$ as in the non-geometric expression.

The conversion is completed by interpreting $f_t=g_{t\ast}f\in\Omega^6(TQ)$ and $\xi_t=\partial_tg_t\circ g_t^{-1}$ as new but enslaved variables. In particular, by \eqref{th:derivative_pullbackbyinverse}, the variable $f_t$ satisfies the Vlasov equation
\begin{align}\label{eq:advection-f}
    (\partial_t+\lie_{\xi_t})f_t=0.
\end{align}
Under the assumption of the enslaved definitions, we can then interpret the action $S_f[g_t,A_t,\phi_t]$ as a functional $\mathfrak{S}[\xi_t,f_t,A_t,\phi_t]$ of the \emph{Eulerian variables}  $(\xi_t,f_t,A_t,\phi_t)$ according to
\begin{align}\label{eq:Low-action-eulerian}
    \mathfrak{S}[\xi_t,f_t,A_t,\phi_t]\equiv 
    & \int\limits_{t_1}^{t_2}\int\limits_{TQ} 
    \big[\iprod_{\xi_t}(\vartheta+e\pi^*A_t) -(K_t+e\pi^*\phi_t)\big]f_t dt + S_{EM}[A_t,\phi_t]
\end{align}
where $f_t \in \Omega^6(TQ)$ is promoted to the set of variables as a dynamical top-form. In what follows, we will be using the kinetic energy functional to denote
\begin{align}
    \mathcal{K}[f_t,E_t,B_t]:=\int\limits_{TQ}K_t f_t = \int\limits_{TQ}K(z,\pi^{\ast}E_t(z), \pi^{\ast}B_t(z))f_t.
\end{align}

The process of switching from the Lagrangian variables to the Eulerian by enslaving the relations between $\xi_t$, $f_t$, and $g_t$ is the basis of Euler-Poincar\'e right-reduction \cite{arnold-1966,Cendra_et_al:1998,Holm_1998}.

\subsection{Constrained variations and Euler-Lagrange conditions}
Hamilton's principle of stationary action applied to \eqref{eq:Low-action-geometric} is equivalent to Hamilton's principle of least action applied to \eqref{eq:Low-action-eulerian} as long as we remember the enslaving relations $\xi_t=\partial_tg_t\circ g_t^{-1}$ and $f_t=g_{t\ast}f$. In practice, these relations have consequences on the type of variations the fields $\xi_t$ and $f_t$ are allowed. 
From \eqref{eq:Low-action-geometric}, one perturbs the one-parameter diffeomorphism $g_t$ to a two-parameter diffeomorphism $g_{t,s}$, the one-form $A_t$ to $A_{t,s}$, and the zero-form $\phi_t$ to $\phi_{t,s}$, and computes the variation of the action in the form
\begin{align}
    \partial_s|_{s=0}S_f[g_{t,s},A_{t,s},\phi_{t,s}]=\delta S_f[\delta g_t, \delta A_t, \delta \phi_t],
\end{align}
where $\delta g_t(z)=\partial_s|_{s=0}g_{t,s}(z)\in T_{g_t(z)}TQ$, $\delta A_t=\partial_s|_{s=0}A_{t,s}\in\Omega^1(Q)$, and $\delta\phi_t=\partial_s|_{s=0}\phi_{t,s}\in\Omega^0(Q)$ are arbitrary but vanishing at $t=t_1$ and $t=t_2$. Then, one requests that the first variation of the action vanishes, in accordance with the Hamilton's principle.

Alternatively, and perhaps more directly, variation of the action can be recorded with the variables $\xi_t$ and $f_t$ by simply letting $\xi_{t,s}=\partial_t g_{t,s}\circ g_{t,s}^{-1}$ and $f_{t,s}=g_{t,s\ast}f$, and writing
\begin{align}
    &\partial_s|_{s=0}S_f[g_{t,s},A_{t,s},\phi_{t,s}]=\partial_s|_{s=0}\mathfrak{S}[\xi_{t,s},f_{t,s},A_{t,s},\phi_{t,s}]=\delta \mathfrak{S}[\delta\xi_t,\delta f_{t},\delta A_{t},\delta\phi_{t}]
\end{align}
as long as the variations of the Eulerian variables respect the relations
\begin{align}
    \delta\xi_t&=\partial_s|_{s=0}(\partial_tg_{t,s}\circ g_{t,s}^{-1})\in\mathfrak{X}(TQ),\\
    \delta f_t&=\partial_s|_{s=0}({g_{t,s}}_\ast f)\in\Omega^6(TQ).
\end{align}
These expressions can be made more transparent by introducing the arbitrary time-dependent vector field $\eta_t=\delta g_t\circ g_t^{-1}\in\mathfrak{X}(TQ)$, which vanishes for $t=t_1$ and $t=t_2$ since $\delta g_t$ does, and by using the Corollary~\ref{th:derivative_pullbackbyinverse} and the Theorem~\ref{th:derivative_of_vector_field} to recover the identities
\begin{align}
    \delta f_t&=-\lie_{\eta_{t}}f_{t},\label{eq:delta-f}\\
    \delta\xi_t&=(\partial_t+\lie_{\xi_t})\eta_t\label{eq:delta-xi}.
\end{align}
Putting the constrained variations to work, we then compute the variation of the action~\eqref{eq:Low-action-eulerian}. After applying the Leibniz rule a couple of times (for both the Lie derivative and the temporal derivative), the result can be expressed as
\begin{align}\label{eq:variation}
    &\delta \mathfrak{S}[(\partial_t+\lie_{\xi_t})\eta_t,-\lie_{\eta_t} f_{t},\delta A_{t},\delta\phi_{t}]\nonumber
    \\
    &=\int\limits_{t_1}^{t_2}\int\limits_{TQ}\left\{
    (\partial_t+\lie_{\xi_t})\left(\iprod_{\eta_t}(\vartheta+e\pi^{\ast}A_t)f_t\right)
    -\lie_{\eta_t}\left[(\iprod_{\xi_t}(\vartheta+e\pi^{\ast}A_t)-(K_t+e\pi^{\ast}\phi_t))f_t\right]\right\}dt \nonumber\\
    &-\int\limits_{t_1}^{t_2}\int\limits_{TQ}\iprod_{\eta_t}[(\partial_t + \iprod_{\xi_t}\extd)(\vartheta+e\pi^{\ast}A_t) + \extd (K_t+e\pi^{\ast}\phi_t)]  f_t dt\nonumber
    \\
    &+\int\limits_{t_1}^{t_2}\int\limits_{TQ}e\left(\iprod_{\xi_t}\pi^{\ast}\delta A_t-\pi^{\ast}\delta\phi_t\right)f_t dt\nonumber
    \\
    &+\int\limits_{t_1}^{t_2}\int\limits_{Q}\left\{\extd\left[\star\left(\varepsilon_0 E_t-\frac{\delta \mathcal{K}}{\delta E_t}\right)\right]\delta\phi_t-\extd\left[\star\left(\varepsilon_0 E_t-\frac{\delta \mathcal{K}}{\delta E_t}\right)\delta\phi_t\right]\right\}dt\nonumber
    \\
    &+\int\limits_{t_1}^{t_2}\int\limits_{Q}\delta A_t\wedge\left[\star\partial_t\left(\varepsilon_0 E_t-\frac{\delta \mathcal{K}}{\delta E_t}\right)-\extd\star\left(\mu_0^{-1}(B_{\text{ext}}+B_t)+\frac{\delta\mathcal{K}}{\delta B_t}\right)\right]dt\nonumber
    \\
    &-\int\limits_{t_1}^{t_2}\int\limits_{Q}\left\{\extd \left[\delta A_t\wedge\star\left(\mu_0^{-1}(B_{\text{ext}}+B_t)+\frac{\delta\mathcal{K}}{\delta B_t}\right)\right]+\partial_t\left[\delta A_t\wedge\star\left(\varepsilon_0 E_t-\frac{\delta \mathcal{K}}{\delta E_t}\right)\right]\right\} dt.
\end{align}
In the above equation, the functional derivatives of the kinetic-energy functional are identified via the relations
\begin{align}
    \partial_s|_{s=0}\mathcal{K}[f_t,E_{t,s},B_t]&=\int\limits_Q\frac{\delta \mathcal{K}}{\delta E_t}\wedge\star\partial_s|_{s=0}E_{t,s}\\
    \partial_s|_{s=0}\mathcal{K}[f_t,E_t,B_{t,s}]&=\int\limits_Q\frac{\delta\mathcal{K}}{\delta B_t}\wedge\star\partial_s|_{s=0}B_{t,s},
\end{align}
These expressions are well defined since we explicitly request the function $K$ not to depend on the derivatives of $E_t$ or $B_t$. 

Since $\partial Q=\emptyset$ and $\partial TQ=\emptyset$, the spatial boundary terms in \eqref{eq:variation} will vanish. Furthermore, since $\eta_{t},\delta A_t,\delta\phi_t$ all vanish at $t=t_1$ and $t=t_2$, also the temporal boundary terms will vanish. For the Hamilton's principle of stationary action to hold, namely that $\delta \mathfrak{S}[(\partial_t+\lie_{\xi_t})\eta_t,-\lie_{\eta_t} f_{t},\delta A_{t},\delta\phi_{t}]=0$ with respect to arbitrary $\eta_t,\delta A_t,\delta\phi_t$, it is enough to request the following Euler-Lagrange conditions for 
the vector field $\xi_t$
\begin{align}\label{eq:Euler-Lagrange-xi}
    \extd(K_t+e\pi^{\ast}\phi_t)+(\partial_t+\iprod_{\xi_t}\extd)(\vartheta+e\pi^{\ast}A_t)=
    \iprod_{\xi_t}(\extd \vartheta + e\pi^\ast B_t) + \extd K_t -  e \pi^\ast E_t =0,
\end{align}
for the magnetic one-form $A_t$ 
\begin{align}\label{eq:Euler-Lagrange-A}
    \int\limits_{TQ}ef_t\iprod_{\xi_t}\pi^{\ast}\delta A_t=\int\limits_{Q}\delta A_t\wedge(\extd\star H_t-\star \partial_t D_t) \quad \Longleftrightarrow \quad \star\partial_t D_t+e\pi_{\ast}(\iprod_{\xi_t}f_t)=\extd\star H_t,
\end{align}
and for the scalar potential $\phi_t$
\begin{align}\label{eq:Euler-Lagrange-phi}
    \int\limits_{TQ}e\pi^*\delta\phi_t f_t&=\int\limits_{Q}\delta\phi_t\extd\star D_t \quad \Longleftrightarrow \quad\extd \star D_t=e\pi_{\ast}(f_t).
\end{align}
Here $\pi_{\ast}(\cdot)$ denotes a fibre integral\footnote{Given a map $h:E\rightarrow P$, fibre integration $h_{\ast}(\cdot)$ satisfies $\int_P \alpha\wedge h_*(\beta) = \int_E h^*\alpha\wedge \beta $. Taking $E=TQ$, $P=Q$, $h=\pi$, $\alpha=\delta A_t$ and $\beta=f_t$, we rewrite $\int_{TQ}\iprod_{\xi_t}\pi^*\delta A_t f_t = \int_{TQ}\pi^*\delta A_t\wedge \iprod_{\xi_t}f_t = \int_{Q}\delta A_t\wedge\pi_{\ast}(\iprod_{\xi_t}f_t)$, where the first step follows because $f_t$ is a top-form and so $\omega\wedge f_t=0$ for any $\omega\in \Omega^k(TQ)$ and because the interior product is an anti-derivation, namely $\iprod_X(\omega\wedge\beta)=\iprod_X\omega\wedge\beta +(-1)^k\omega\wedge \iprod_X \beta$.} from $TQ$ down to $Q$, and the one-form $D_t\in\Omega^1(Q)$ and the two-form $H_t\in\Omega^2(Q)$ have been introduced to denote the displacement and magnetising fields
\begin{align}
    D_t&=\varepsilon_0 E_t-\frac{\delta\mathcal{K}}{\delta E_t},\\
    H_t&=\mu_0^{-1}(B_{\text{ext}}+B_t)+\frac{\delta\mathcal{K}}{\delta B_t}.
\end{align}
\rev{Note that equation \eqref{eq:Euler-Lagrange-xi}, determining the vector field $\xi_t$, effectively provides the characteristics of single-particle motion for the Vlasov equation while the equations \eqref{eq:Euler-Lagrange-A} and \eqref{eq:Euler-Lagrange-phi} are the geometric versions of the Amp\`ere-Maxwell equation and the Gauss's law for the electric field including polarization and magnetization effects.}
\section{Noether equations for spatial isometries and time translations}
To study the effects of spatial isometries\footnote{Isometries on a manifold $M$ are distance preserving diffeomorphism. On $\mathbb{R}^3$ these include constant translations and rotations. The pullbacks of isometries commute with the Hodge operator $\star$.} and time translations, we will construct a new functional that is obtained from the action functional evaluated over not the whole of $Q$ and $TQ$ but the subsets $U\subseteq Q$ and $TU=\bigcup\{(x,v_x)|x\in U,v_x\in T_xQ\}\subseteq TQ$. In effect, this new functional can then be treated as to parametrically depend on the domain $U$ and the temporal end-points $t_1$ and $t_2$. The new functional we introduce is given by
\begin{align}\label{eq:action-reparametrization}
    \mathfrak{S}_{U,t_1,t_2}[\xi_t,f_t,A_t,\phi_t]=&\int\limits_{t_1}^{t_2}\int\limits_{TU}\iprod_{\xi_t}\vartheta f_tdt-\int\limits_{t_1}^{t_2}\mathcal{K}_{TU}[f_t,E_t,B_t]dt+\int\limits_{t_1}^{t_2}\int\limits_{TU}\big(e\iprod_{\xi_t}\pi^{\ast}A_t-e\pi^{\ast}\phi_t\big)f_tdt\nonumber
    \\&+\int\limits_{t_1}^{t_2}\int\limits_{U}\frac{1}{2}\big(\varepsilon_0 E_t\wedge\star E_t-\mu_0^{-1}(B_{\text{ext}}+B_t)\wedge\star(B_{\text{ext}}+B_t)\big)dt,
\end{align}
where the modified kinetic energy functional is defined in the natural way
\begin{align}
    \mathcal{K}_{TU}[f_t,E_t,B_t]:=\int\limits_{TU}K_t f_t=\int\limits_{TU}K(z,\pi^{\ast}E_t(z), \pi^{\ast}B_t(z))f_t.
\end{align}
Trivially, if we choose $U=Q$, we obtain the original action. 

A few remarks are in order here. In what follows, the functional \eqref{eq:action-reparametrization} will be varied and the functional derivatives of $\mathcal{K}_{TU}$ used. This might raise some questions since no specific form of the function $K_t$ is given yet. Specifically, one could question whether the functional derivatives $\delta \mathcal{K}_{TU}/\delta E_t$ and $\delta\mathcal{K}_{TU}/\delta B_t$ exists at all with respect to an arbitrary domain $U$. This small curiosity was the reason why we restricted our discussion to such $K_t$ which do not depend on the derivatives of $E_t$ or $B_t$. Then the functional derivatives $\delta \mathcal{K}_{TU}/\delta E_t$ and $\delta \mathcal{K}_{TU}/\delta B_t$ are not only well defined but are, in fact, equal to the functional derivatives of $\mathcal{K}$.

\subsection{Spatial isometries}
The idea in analysing symmetries related to spatial isometries is to introduce a one-parameter isometry $\psi_s\in\text{Diff}(Q)$ with $\psi_0=\text{id}$ and its lift $\Psi_s\in\text{Diff}(TQ)$ with $\Psi_0=\text{id}$. The lift in our context means that $\Psi_s$ is required to satisfy $\pi\circ\Psi_s=\psi_s\circ \pi$. Consequently, there will be the vector fields $X=\partial_s|_{s=0}\psi_s\circ \psi_0^{-1}$ and $\widetilde{X}=\partial_s|_{s=0}\Psi_s\circ \Psi_0^{-1}$ which act as the infinitesimal generators for $\psi_s$ and $\Psi_s$ respectively, and are $\pi$-related, i.e., $T\pi\circ \widetilde{X}=X\circ\pi$, and it can be shown that $\Psi_{s\ast}\pi^\ast \alpha = \pi^\ast \psi_{s\ast}\alpha$ for any $\alpha\in \Omega^k(Q)$. Furthermore, since $TU$ is locally $U\times\mathbb{R}^3$, we have that $\text{Im}_{\Psi_s}(TU)=T\text{Im}_{\psi_s}(U)$. With these definitions in mind, one performs a coordinate transformation, acting with $\Psi_s$ on the $TU$ part and with $\psi_s$ on the $U$ part of \eqref{eq:action-reparametrization}, and obtains
\begin{align}\label{eq:spatial-coordinate-transform}
    \mathfrak{S}_{U,t_1,t_2}[\xi_t,f_t,A_t,\phi_t]=&\mathfrak{S}_{\text{Im}_{\psi_s}(U),t_1,t_2}[\Psi_{s\ast}\xi_t,\Psi_{s\ast}f_t,\psi_{s\ast}A_t,\psi_{s\ast}\phi_t]
    \nonumber
    \\
    &+\int\limits_{t_1}^{t_2}\int\limits_{\text{Im}_{\Psi_s}(TU)}\iprod_{\Psi_{s\ast}\xi_t}(\Psi_{s\ast}-\text{id})\vartheta\, \Psi_{s\ast}f_tdt
    \nonumber
    \\
    &-\int\limits_{t_1}^{t_2}\int\limits_{\text{Im}_{\Psi_s}(TU)}[\Psi_{s\ast}K_t-K(z,\Psi_{s\ast}\pi^{\ast}E_t(z),\Psi_{s\ast}\pi^{\ast}B_t(z))]\Psi_{s\ast}f_tdt
    \nonumber
    \\&-\int\limits_{t_1}^{t_2}\int\limits_{\text{Im}_{\psi_s}(U)}\mu_0^{-1}(\psi_{s\ast}-\text{id})B_{\text{ext}}\wedge\star(B_{\text{ext}}+\psi_{s\ast}B_t)dt\nonumber
    \\&-\int\limits_{t_1}^{t_2}\int\limits_{\text{Im}_{\psi_s}(U)}\frac{1}{2}\mu_0^{-1}(\psi_{s\ast}-\text{id})B_{\text{ext}}\wedge\star(\psi_{s\ast}-\text{id})B_{\text{ext}}dt.
\end{align}

If some specific isometry $\psi_s$ and its lift $\Psi_s$ are to generate a symmetry in the sense that
\begin{align}\label{eq:spatial-symmetry}
    \mathfrak{S}_{U,t_1,t_2}[\xi_t,f_t,A_t,\phi_t]=\mathfrak{S}_{\text{Im}_{\psi_s}(U),t_1,t_2}[\Psi_{s\ast}\xi_t,\Psi_{s\ast}f_t,\psi_{s\ast}A_t,\psi_{s\ast}\phi_t],
\end{align}
then this isometry and its lift have to satisfy the conditions
\begin{subequations}
\label{eq:spatial-symmetry-conditions}
\begin{align}
    \psi_{s\ast}B_{\text{ext}}&=B_{\text{ext}},\\
    \Psi_{s\ast}\vartheta&=\vartheta,\\
    K(\Psi_s^{-1}(z),\pi^{\ast}E_t(\Psi_s^{-1}(z)),\pi^{\ast}B_t(\Psi_s^{-1}(z)))&=K(z,\Psi_{s\ast}\pi^{\ast}E_t(z),\Psi_{s\ast}\pi^{\ast}B_t(z)).
\end{align}
\end{subequations}
If the conditions \eqref{eq:spatial-symmetry-conditions} are satisfied, the existence of a local conservation law will be guaranteed by Noether's first theorem. These are the {\it strong} conditions for a conservation law to exists. There are also weaker conditions, which we will discuss shortly.

To extract the local conservation law, the expression \eqref{eq:spatial-symmetry} will be differentiated with respect to $s$ at $s=0$ and evaluated {\it on-shell}, i.e., the Euler-Lagrange conditions required to hold. This provides, subject to the symmetry conditions, that
\begin{align}
    0=&\partial_s|_{s=0}\mathfrak{S}_{\text{Im}_{\psi_s}(U),t_1,t_2}[\xi_t,f_t,A_t,\phi_t]+\delta \mathfrak{S}_{U,t_1,t_2}[-\lie_{\widetilde{X}}\xi_t,-\lie_{\widetilde{X}}f_t,-\lie_{X}A_t,-\lie_{X}\phi_t].
\end{align}
Applying the fundamental theorem of calculus, the first term can be evaluated immediately
\begin{align}\label{eq:spatial-boundary-derivative}
    &\partial_s|_{s=0}\mathfrak{S}_{\text{Im}_{\psi_s}(U),t_1,t_2}[\xi_t,f_t,A_t,\phi_t]\nonumber
    \\
    &=\int\limits_{t_1}^{t_2}\int\limits_{TU}\lie_{\widetilde{X}}\left(\iprod_{\xi_t}\vartheta f_t-K_tf_t+(e\iprod_{\xi_t}\pi^{\ast}A_t-e\pi^{\ast}\phi_t)f_t\right)dt\nonumber
    \\&\qquad+\int\limits_{t_1}^{t_2}\int\limits_{U}\frac{1}{2}\lie_{X}\big(\varepsilon_0 E_t\wedge\star E_t-\mu_0^{-1}(B_{\text{ext}}+B_t)\wedge\star(B_{\text{ext}}+B_t)\big)dt.
\end{align}
To evaluate the term $\delta \mathfrak{S}_{U,t_1,t_2}[-\lie_{\widetilde{X}}\xi_t,-\lie_{\widetilde{X}}f_t,-\lie_{X}A_t,-\lie_{X}\phi_t]$ on-shell, we use the fact that $X$ and $\widetilde{X}$ are both independent of time $t$ so that $-\lie_{\widetilde{X}}\xi_t=(\partial_t+\lie_{\xi_t})\widetilde{X}$. This helps us recognize that the term can be evaluated as a special case of $\eqref{eq:variation}$ with $\eta_t=\widetilde{X}$, $\delta A_t=-\lie_XA_t=-\iprod_XB_t-\extd(\iprod_XA_t)$ and $\delta\phi_t=-\lie_X\phi_t=\iprod_XE_t+\partial_t(\iprod_XA_t)$, now only evaluated over $U$ and $TU$ instead of $Q$ and $TQ$. This means that when the Euler-Lagrange conditions are implied, only the boundary terms, that vanish in \eqref{eq:variation}, will remain. It is then a straightforward task to compute the on-shell variation \rev{(see \ref{sec:on-shell-variation-spatial-details} for details)}
\begin{align}\label{eq:on-shell-variation-spatial}
    &\delta \mathfrak{S}_{U,t_1,t_2}[(\partial_t+\lie_{\xi_t})\widetilde{X},-\lie_{\widetilde{X}} f_{t},-\iprod_XB_t-\extd(\iprod_XA_t),\iprod_XE_t+\partial_t(\iprod_XA_t)]\nonumber
    \\
    &=\int\limits_{t_1}^{t_2}\int\limits_{TU}(\partial_t+\lie_{\xi_t})\big(f_t\iprod_{\widetilde{X}}\vartheta\big)-\lie_{\widetilde{X}}\Big[f_t\iprod_{\xi_t}(\vartheta+e\pi^{\ast}A_t)-(K_t+e\pi^{\ast}\phi_t)f_t\Big]dt\nonumber
    \\
    &\qquad+\int\limits_{t_1}^{t_2}\int\limits_{U}\Big[\extd (\iprod_XB_t\wedge\star H_t-\star D_t\iprod_XE_t)+\partial_t(\iprod_XB_t\wedge\star D_t)\Big] dt.
\end{align}
Finally, combining the on-shell variation \eqref{eq:on-shell-variation-spatial} with the expression \eqref{eq:spatial-boundary-derivative}, and requesting the result to be true with respect to arbitrary domain $U$, a local conservation law is obtained
\begin{align}\label{eq:spatial-conservation-law}
    &\partial_t(\pi_{\ast}(f_t\iprod_{\widetilde{X}}\vartheta)+\iprod_XB_t\wedge\star D_t)+\pi_{\ast}(\lie_{\xi_t}(f_t\iprod_{\widetilde{X}}\vartheta))-\extd(\rev{B_t\iprod_X\star H_t}+\star D_t\iprod_XE_t)
    \nonumber
    \\
    &\rev{+\lie_X(B_t\wedge\star H_t)+\frac{1}{2}\lie_{X}\big(\varepsilon_0 E_t\wedge\star E_t-\mu_0^{-1}(B_{\text{ext}}+B_t)\wedge\star(B_{\text{ext}}+B_t)\big)}=0.
\end{align}

At this point, we remind that for this equation to hold, the symmetry conditions \eqref{eq:spatial-symmetry-conditions} must be true. In case the isometry does not satisfy the symmetry conditions, one may still differentiate \eqref{eq:spatial-coordinate-transform} with respect to $s$ at $s=0$ and account for the remaining volumetric terms. In that case, equation \eqref{eq:spatial-conservation-law} would be modified by a volumetric source term $S$ appearing on the right, the source term being
\begin{align}\label{eq:source}
    S=&\pi_{\ast}(\iprod_{\xi_t}\lie_{\widetilde{X}}\vartheta f_t-\lie_{\widetilde{X}}K_tf_t)+\frac{\delta \mathcal{K}}{\delta E_t}\wedge\star\lie_XE_t\nonumber
    \\&+\frac{\delta \mathcal{K}}{\delta B_t}\wedge\star\lie_XB_t+\mu_0^{-1}\lie_XB_{\text{ext}}\wedge\star(B_{\text{ext}}+B_t).
\end{align}
From this expression, we see that the {\it weak} condition for a conservation law to exist is that this source term vanishes, given the Euler-Lagrange conditions. Alternatively, the source term can be used to investigate the momentum balance of the system in directions other than the obvious symmetry direction of the external magnetic field.

\subsection{Constant translations in time}
Analysing constant translations in time is simpler than the analysis of spatial isometries for there is no need to consider lifts or diffeomorphisms at all. Since the action does not have parametric dependencies on time, i.e., $\partial_t\vartheta=0$ and the function $K_t$ depends on time only via $E_t$ and $B_t$, we immediately obtain for any constant $T$ the following, strong symmetry condition
\begin{align}\label{eq:temporal-symmetry}
    \mathfrak{S}_{U,t_1,t_2}[\xi_t,f_t,A_t,\phi_t]=\mathfrak{S}_{U,t_1+T,t_2+T}[\xi_{t-T},f_{t-T},A_{t-T},\phi_{t-T}]
\end{align}
and there will be a related conservation law guaranteed by Noether's first theorem.

To extract the conservation law, we proceed as with the spatial isometries, differentiating \eqref{eq:temporal-symmetry} with respect to $T$ at $T=0$:
\begin{align}
    0=&\partial_T|_{T=0}\mathfrak{S}_{U,t_1+T,t_2+T}[\xi_t,f_t,A_t,\phi_t]+\delta \mathfrak{S}_{U,t_1,t_2}[-\partial_t\xi_t,-\partial_tf_t,-\partial_tA_t,-\partial_t\phi_t].
\end{align}
Using again the fundamental theorem of calculus, the first term is straightforward to evaluate
\begin{align}\label{eq:temporal-boundary-derivative}
    &\partial_T|_{T=0}\mathfrak{S}_{U,t_1+T,t_2+T}[\xi_t,f_t,A_t,\phi_t]\nonumber
    \\
    &=\int\limits_{t_1}^{t_2}\int\limits_{TU}\partial_t\big(f_t\iprod_{\xi_t}(\vartheta+e\pi^{\ast}A_t)-(K_t+e\pi^{\ast}\phi_t)f_t\big)dt\nonumber
    \\&\qquad+\int\limits_{t_1}^{t_2}\int\limits_{U}\frac{1}{2}\partial_t\big(\varepsilon_0 E_t\wedge\star E_t-\mu_0^{-1}(B_{\text{ext}}+B_t)\wedge\star(B_{\text{ext}}+B_t)\big)dt.
\end{align}
To evaluate the second term, we apply a trick similar to what we used in analysing the spatial isometries: we re-express $-\partial_t\xi_t=(\partial_t+\lie_{\xi_t})(-\xi_t)$ and $-\partial_t f_t=-\lie_{-\xi_t}f_t$. This observation then helps us identify that $\delta \mathfrak{S}_{U,t_1,t_2}[-\partial_t\xi_t,-\partial_tf_t,-\partial_tA_t,-\partial_t\phi_t]$ is effectively a special case of \eqref{eq:variation} with $\eta_t=-\xi_t$, $\delta A_t=E_t+\extd\phi_t$, and $\delta\phi_t=-\partial_t\phi_t$, now only evaluated over $U$ and $TU$ instead of $Q$ and $TQ$. Direct substitution then provides the on-shell variation \rev{(see \ref{sec:on-shell-variation-temporal-details} for details)}
\begin{align}\label{eq:on-shell-variation-temporal}
    &\delta \mathfrak{S}_{U,t_1,t_2}[(\partial_t+\lie_{\xi_t})(-\xi_t),-\lie_{-\xi_t} f_{t},E_t+\extd\phi_t,-\partial_t\phi_t]\nonumber
    \\
    &=-\int\limits_{t_1}^{t_2}\int\limits_{TU}\partial_t\big(f_t\iprod_{\xi_t}(\vartheta+e\pi^{\ast}A_t)-(K_t+e\pi^{\ast}\phi_t)f_t\big)dt\nonumber
    \\
    &\qquad-\int\limits_{t_1}^{t_2}\int\limits_{U}\Big[\extd (E_t\wedge\star H_t)+\partial_t(E_t\wedge\star D_t)\Big] dt
    \nonumber
    \\
    &\qquad-\int\limits_{t_1}^{t_2}\int\limits_{TU}(\partial_t+\lie_{\xi_t})(f_tK)dt.
\end{align}
Putting everything together by summing \eqref{eq:temporal-boundary-derivative} and \eqref{eq:on-shell-variation-temporal}, and noting that the domain $U$ is arbitrary, we obtain the local conservation law for the energy density
\begin{align}\label{eq:temporal-conservation-law}
    &\partial_t\big(\pi_{\ast}(f_tK_t)+E_t\wedge\star D_t-\tfrac{1}{2}\varepsilon_0 E_t\wedge\star E_t+\tfrac{1}{2}\mu_0^{-1}(B_{\text{ext}}+B_t)\wedge\star(B_{\text{ext}}+B_t)\big)\nonumber
    \\
    &\qquad+\pi_{\ast}(\lie_{\xi_t}(f_tK_t))+\extd (E_t\wedge\star H_t)=0.
\end{align}

\section{Example applications}
Explicitly, we shall consider two models, namely the full-particle Vlasov-Maxwell \rev{in a background magnetic field} and the drift-kinetic Vlasov-Maxwell that is obtainable as the long-wave-length limit of the gyrokinetic Vlasov-Maxwell system \cite{Burby-Brizard:2019PhLA}. For the external magnetic field, we shall consider the axially symmetric, time-independent magnetic field often encountered in a tokamak. In cylindrical coordinates $(R,\varphi,z)$, the vector-calculus representation of such field is given by
\begin{align}
    \bm{B}_{\text{ext}}=G(R,z)\nabla\varphi+\nabla\Psi(R,z)\times\nabla\varphi.
\end{align}
This field admits a rotational symmetry with respect to an isometry $\psi_s$ and the related vector field $X=\partial_s|_{s=0}\psi_s\circ\psi_0^{-1}$, that are defined via
\begin{subequations}
    \label{eq:isometry}
    \begin{align}
        &\psi_s(R,\varphi,z)=(R,\varphi+s,z),\\
        &X=\hat{\bm{z}}\times\bm{x}\cdot\nabla=\bm{e}_{\varphi}\cdot\nabla=\partial_{\varphi}.
    \end{align}
\end{subequations}
Expressed mathematically, the symmetry exists in the sense of
\begin{align}
    \psi_{s\ast}B_{\text{ext}}&=B_{\text{ext}}\\
    \psi_{s\ast}A_{\text{ext}}&=A_{\text{ext}}.
\end{align}
which, in coordinates and in differential sense, means that $\partial_{\varphi}\bm{B}_{\text{ext}}=\hat{\bm{z}}\times\bm{B}_{\text{ext}}$ and $\partial_{\varphi}\bm{A}_{\text{ext}}=\hat{\bm{z}}\times\bm{A}_{\text{ext}}$. Naturally, since this field admits only a rotational symmetry, there will be no conservation law for linear momentum density. The conservation law for linear momentum density would require a translational symmetry in $B_{\text{ext}}$, a case which we leave as an exercise for an interested reader to verify with the machinery we have presented in the previous section.

And since we are merely applying the machinery derived earlier, we will perform the computations in this section in coordinates and provide the results in terms of regular vector calculus. This choice will hopefully make these example computations approachable to a larger audience.

\subsection{Classic full-particle Vlasov-Maxwell}
In the classic Vlasov-Maxwell system, the kinetic energy of a particle depends only on the velocity coordinate $\bm{v}$. Considering the possibility of the external axially symmetric magnetic field, the one-form $\vartheta$ and the kinetic energy function $K$ are then given by the coordinate expressions
\begin{align}
    \vartheta&=e\bm{A}_{\text{ext}}\cdot\extd \bm{x}+m\bm{v}\cdot\extd \bm{x},\\
    K_t&=\frac{1}{2}m|\bm{v}|^2.
\end{align}
In component form, the Euler-Lagrange condition \eqref{eq:Euler-Lagrange-xi} for $\xi_t$ is given by
\begin{align}
    m\bm{v}\cdot\extd\bm{v}-e(\bm{E}_t+\bm{\xi}^x_t\times(\bm{B}_{\text{ext}}+\bm{B}_t))\cdot\extd\bm{x}+m\bm{\xi}^{v}_t\cdot\extd\bm{x}-m\bm{\xi}^x_t\cdot\extd\bm{v}=0,
\end{align}
which is straightforward to invert for the components
\begin{align}
    \bm{\xi}_t^x&=\bm{v},\\
    \bm{\xi}_t^v&=\frac{e}{m}(\bm{E}_t+\bm{v}\times(\bm{B}_{\text{ext}}+\bm{B}_t)).
\end{align}
Furthermore, since the energy function $K_t$ is now entirely independent of the electric and magnetic field, the components of the one-form $D_t$ and the two-form $H_t$ are given by $\bm{D}_t=\varepsilon_0\bm{E}_t$ and $\bm{H}_t=\mu_0^{-1}(\bm{B}_{\text{ext}}+\bm{B}_t)$. The equations \eqref{eq:Euler-Lagrange-A} and \eqref{eq:Euler-Lagrange-phi} then provide the standard Gauss's and Faraday's laws
\begin{align}
    \varepsilon_0\partial_t\bm{E}_t-\mu_0^{-1}\nabla\times(\bm{B}_{\text{ext}}+\bm{B}_t)+\bm{j}_t&=0,\\
    \varepsilon_0\nabla\cdot\bm{E}_t-\varrho_t&=0,
\end{align}
with the current and charge densities computed from the density $f_t=F_td^3\bm{x}d^3\bm{v}$ as the velocity space integrals
\begin{align}
    \bm{j}_t&=e\int \bm{\xi}_t^xF_t d^3\bm{v},\\
    \varrho_t&=e\int F_t d^3\bm{v}.
\end{align}
Finally, the Vlasov equation is obtained from the enslaved advection condition
\begin{align}
    (\partial_t+\lie_{\xi_t})f_t=(\partial_tF_t+\partial_{z^{\alpha}}(\xi_t^{\alpha}F_t))d^6\bm{z}=0.
\end{align}

To check the symmetry conditions \eqref{eq:spatial-symmetry-conditions}, we use their differential form (differentiation with respect to $s$) and consider the tangential lift $\Psi_s(x,v)=(\psi_s(x),\psi_s(v))$ with the corresponding vector field given in components according to
\begin{align}
    \tilde{X}=\hat{\bm{z}}\times\bm{x}\cdot\nabla+\hat{\bm{z}}\times\bm{v}\cdot\partial/\partial_{\bm{v}}
\end{align}
It is then a straightforward to verify that
\begin{align}
    &\lie_{\widetilde{X}}(\bm{v}\cdot\extd\bm{x})=0,\\    &\lie_{\widetilde{X}}\tfrac{1}{2}|\bm{v}|^2=0,
\end{align}
Obtaining the associated conservation law is then a matter of translating \eqref{eq:spatial-conservation-law} to the language of ordinary vector calculus. The result, the conservation law for the angular momentum density, becomes
\begin{align}
    &\partial_t\Big(\int F_t (m\bm{v}+e\bm{A}_{\text{ext}})\cdot\bm{e}_{\varphi}d^3\bm{v}\rev{+\varepsilon_0\bm{E}_t\times\bm{B}_t\cdot\bm{e}_{\varphi}}\Big)\nonumber
    \\
    &+\nabla\cdot\Big(\int \bm{v}F_t (m\bm{v}+e\bm{A}_{\text{ext}})\cdot\bm{e}_{\varphi}d^3\bm{v}+\tfrac{1}{2}\varepsilon_0|\bm{E}_t|^2\bm{e}_{\varphi}-\tfrac{1}{2}\mu_0^{-1}|\bm{B}_{\text{ext}}+\bm{B}_t|^2\bm{e}_{\varphi}\nonumber
    \\
    &\qquad-\varepsilon_0\bm{E}_t\bm{E}_t\cdot\bm{e}_{\varphi}-\mu_0^{-1}\bm{B}_t(\bm{B}_{\text{ext}}+\bm{B}_t)\cdot\bm{e}_{\varphi}+\mu_0^{-1}\bm{B}_t\cdot(\bm{B}_{\text{ext}}+\bm{B}_t)\bm{e}_{\varphi}\Big)=0.
\end{align}
In a similar manner, we translate \eqref{eq:temporal-conservation-law} to vector calculus and write down the conservation law for energy density
\begin{align}
    &\partial_t\Big(\int \tfrac{1}{2}m|\bm{v}|^2 F_td^3\bm{v}+\tfrac{1}{2}\varepsilon_0|E_t|^2+\tfrac{1}{2}\mu_0^{-1}|\bm{B}_{\text{ext}}+\bm{B}_t|^2\Big)\nonumber
    \\
    &+\nabla\cdot\Big(\int\tfrac{1}{2}m|\bm{v}|^2\bm{v}F_td^3\bm{v}+\mu_0^{-1}\bm{E}_t\times(\bm{B}_{\text{ext}}+\bm{B}_t)\Big)=0.
\end{align}

\subsection{Drift-kinetic Vlasov-Maxwell}
In the drift-kinetic Vlasov-Maxwell, the one-form $\vartheta$ and the kinetic energy $K_t$ are given by the coordinate expressions
\begin{align}
    \vartheta&=e\bm{A}_{\text{ext}}\cdot\extd\bm{x}+mv_{\parallel}\bm{b}_{\text{ext}}\cdot \extd\bm{x}+(m/e)\mu \extd\theta\\
    K_t&=\frac{1}{2}mv_{\parallel}^2+\mu |\bm{B}_{\text{ext}}|\left(1+\frac{\bm{b}_{\text{ext}}\cdot\bm{B}_t}{|\bm{B}_{\text{ext}}|}+\frac{|\bm{B}_{t\perp}|^2}{2|\bm{B}_{\text{ext}}|^2}\right)-\frac{m}{2|\bm{B}_{\text{ext}}|^2}|\bm{E}_{t\perp}+v_{\parallel}\bm{b}_{\text{ext}}\times\bm{B}_t|^2
\end{align}
with the subscript $\perp$ referring to dot product with the dyad $\mathbf{1}_{\perp}=\mathbf{1}-\bm{b}_{\text{ext}}\bm{b}_{\text{ext}}$ and $\bm{b}_{\text{ext}}=\bm{B}_{\text{ext}}/|\bm{B}_{\text{ext}}|$ is the unit vector in the direction of the external magnetic field. The Euler-Lagrange condition \eqref{eq:Euler-Lagrange-xi} for the vector field $\xi_t$ gives 
\begin{align}
    &\nabla K_t\cdot\extd\bm{x}+\partial_{v_{\parallel}}K_t\extd v_{\parallel}+\partial_{\mu}K_t\extd \mu-e(\bm{E}_t+\bm{\xi}_t^x\times(\bm{B}_t+\bm{B}_{\text{ext}}+(m/e)v_{\parallel}\nabla\times\bm{b}_{\text{ext}}))\cdot\extd\bm{x}\nonumber
    \\
    &+(m/e)(\xi_t^{\mu}\extd \theta-\xi_t^{\theta}\extd\mu) +\xi^{v_{\parallel}}_tm\bm{b}_{\text{ext}}\cdot\extd\bm{x}-m\bm{b}_{\text{ext}}\cdot\bm{\xi}_t^x \extd v_{\parallel}=0.
\end{align}
From this expression, we invert for the components
\begin{align}
    \bm{\xi}_t^x&=\frac{\partial_{v_{\parallel}}K_t}{m}\frac{\bm{B}_t^{\star}}{\bm{b}_{\text{ext}}\cdot\bm{B}_t^{\star}}+\frac{(e\bm{E}_t-\nabla K_t)\times\bm{b}_{\text{ext}}}{e\bm{b}_{\text{ext}}\cdot\bm{B}_t^{\star}},\\
    \xi_t^{v_{\parallel}}&=\frac{\bm{B}_t^{\star}\cdot(e\bm{E}_t-\nabla K_t)}{m\bm{b}_{\text{ext}}\cdot\bm{B}^{\star}_t},\\
    \xi_t^{\mu}&=0,\\
    \xi_t^{\theta}&=\frac{e}{m}\frac{\partial K_t}{\partial\mu},
\end{align}
where $\bm{B}_t^{\star}=\bm{B}_t+\bm{B}_{\text{ext}}+(m/e)v_{\parallel}\nabla\times\bm{b}_{\text{ext}}$. The Euler-Lagrange conditions for $A_t$ \eqref{eq:Euler-Lagrange-A} and $\phi_t$ \eqref{eq:Euler-Lagrange-phi} provide 
\begin{align}
    \partial_t\bm{D}_t-\nabla\times\bm{H}_t+\bm{j}_t&=0,\\
    \nabla\cdot\bm{D}_t-\varrho_t,&=0
\end{align}
where the macroscopic fields $\bm{D}_t$ and $\bm{B}_t$ and the sources $\bm{j}_t$ and $\varrho_t$ are defined as 
\begin{align}
    \bm{D}_t&=\varepsilon_0\bm{E}_t-\int \partial_{\bm{E}_t}K_t F_t dv_{\parallel}d\mu d\theta,\\
    \bm{H}_t&=\mu_0^{-1}(\bm{B}_{\text{ext}}+\bm{B}_t)+\int \partial_{\bm{B}_t}K_t F_tdv_{\parallel}d\mu d\theta,\\
    \bm{j}_t&=\int e\bm{\xi}^x_t F_t dv_{\parallel}d\mu d\theta,\\
    \varrho_t&=\int e F_t dv_{\parallel}d\mu d\theta.
\end{align}
The Vlasov equation is obtained, as previously, from the enslaved advection condition
\begin{align}
    (\partial_t+\lie_{\xi_t})f_t=(\partial_tF_t+\partial_{z^{\alpha}}(\xi_t^{\alpha}F_t))d^6\bm{z}=0.
\end{align}

To check the symmetry conditions \eqref{eq:spatial-symmetry-conditions}, we again use their differential form and consider the tangential lift $\Psi_s(x,v)=(\psi_s(x),\psi_s(v))$. Now the component form of the vector field $\widetilde{X}$ is, however, given by the expression
\begin{align}
    \tilde{X}=\hat{\bm{z}}\times\bm{x}\cdot\nabla,
\end{align}
which follows from the fact that rotating the guiding-center-particle velocity along the symmetry direction of the external magnetic field does not change the values of the coordinates $v_{\parallel}$, $\mu$, or $\theta$ as they are defined locally with respect to the direction and magnitude of the external magnetic field.
It is then a straightforward computation to verify the infinitesimal forms of the symmetry conditions, namely that
\begin{align}
    &\lie_{\widetilde{X}}\vartheta=e(\bm{A}_{\text{ext}}^{\star}\times\nabla\times\bm{e}_{\varphi}+\bm{e}_{\varphi}\cdot\nabla \bm{A}_{\text{ext}}^{\star}+\bm{A}_{\text{ext}}^{\star}\cdot\nabla\bm{e}_{\varphi})\cdot d\bm{x}=0,
    \\ 
    &\partial_{\varphi}K_t+\partial_{\bm{B}_t}K_t\cdot(\hat{\bm{z}}\times\bm{B}_t-\partial_{\varphi}\bm{B}_t)+\partial_{\bm{E}_t}K_t\cdot(\hat{\bm{z}}\times\bm{E}_t-\partial_{\varphi}\bm{E}_t)=0,
\end{align}
where $e\bm{A}^{\star}_{\text{ext}}=e\bm{A}_{\text{ext}}+mv_{\parallel}\bm{b}_{\text{ext}}$.
The conservation law for angular momentum density is then obtained after translating \eqref{eq:spatial-conservation-law} to the language of ordinary vector calculus. The result is
\begin{align}
    &\partial_t\Big(\int F_t (e\bm{A}_{\text{ext}}+mv_{\parallel}\bm{b}_{\text{ext}})\cdot\bm{e}_{\varphi}dv_\parallel d\mu d\theta\rev{+\bm{D}_t\times\bm{B}_t\cdot\bm{e}_{\varphi}}\Big)\nonumber
    \\
    &+\nabla\cdot\Big(\int\bm{\xi}_t^xF_t (e\bm{A}_{\text{ext}}+mv_{\parallel}\bm{b}_{\text{ext}})\cdot\bm{e}_{\varphi}dv_\parallel d\mu d\theta+\tfrac{1}{2}\varepsilon_0|\bm{E}_t|^2\bm{e}_{\varphi}\nonumber
    \\
    &\qquad-\tfrac{1}{2}\mu_0^{-1}|\bm{B}_{\text{ext}}+\bm{B}_t|^2\bm{e}_{\varphi}-\bm{D}_t\bm{E}_t\cdot\bm{e}_{\varphi}-\bm{B}_t\bm{H}_t\cdot\bm{e}_{\varphi}+\bm{B}_t\cdot\bm{H}_t\bm{e}_{\varphi}\Big)=0.
\end{align}
In a similar manner, we translate \eqref{eq:temporal-conservation-law} to vector calculus and obtain the conservation law for energy density
\begin{align}
    &\partial_t\Big(\int K_tF_tdv_\parallel d\mu d\theta+\bm{D}_t\cdot\bm{E}_t-\tfrac{1}{2}\varepsilon_0|\bm{E}_t|^2+\tfrac{1}{2}\mu_0^{-1}|\bm{B}_0+\bm{B}_t|^2\Big)\nonumber\\
    &+\nabla\cdot\Big(\int \bm{\xi}_t^x K_tF_tdv_\parallel d\mu d\theta+\bm{E}_t\times\bm{H}_t\Big)=0.
\end{align}

\section{Summary}
In this paper, we have reviewed the geometric interpretation of the Euler-Poincar\'e formulation for the purposes of applying it to Vlasov-Maxwell-type systems encountered in the kinetic theory of plasmas, and explained how the possible conservation laws related to constant rotations and translations in space and translations in time can be obtained in an algorithmic manner. After the rather mathematical exposition, two explicit examples were given---the full-particle \rev{Vlasov-Maxwell in an axially symmetric tokamak-like background magnetic field} and the drift-kinetic Vlasov-Maxwell \rev{obtained as the long-wave-length limit of the gyrokinetic model \cite{Burby-Brizard:2019PhLA}}---with the results being translated to the language of regular vector calculus in the end. We hope that readers would find the demonstrative calculations helpful in their own endeavours and that the explicit demonstrations of the geometric take on the Euler-Poincar\'e methodology would help unmask its potential to the plasma physics community.

\begin{acknowledgments}
    Research presented in this article was supported by the Academy of Finland grant no. 315278, by the Los Alamos National Laboratory LDRD program under project number 20180756PRD4, and by the National Science Foundation under contract No. PHY-1805164. Any subjective views or opinions expressed herein do not necessarily represent the views of the Academy of Finland, Aalto University, Los Alamos National Laboratory, The University of Western Australia, or the National Science Foundation. 
\end{acknowledgments}

\appendix
\section{Certain useful identities}
It is useful to review a few identities in order to understand the origins of the constrained variations in the Euler-Poincar\'e formalism. Parts of this material are covered in, e.g., Ref. \cite{Holm_1998} Section 6, where also the general theory of Euler-Poincar\'e reduction is presented. We first recall some basic definitions:

\begin{definition}[Tangent map]
Given a smooth map $\varphi:U\to V$ between open subset $U\subseteq \R^m$ and $V\subseteq \R^n$, the differential of $\varphi$ at point $x\in U$, $T_x\varphi:\R^m \to \R^n$ is the unique linear map such that $\lim\limits_{||v||_{\R^m}\to 0} ||\varphi(x+v)-\varphi(x)-T_x\varphi(v)||_{\R^n}/||v||_{\R^m}=0$. This concept generalises to smooth maps $\varphi:M\to N$ between any smooth manifolds $M$ and $N$, defining what is called the tangent map (or pushforward) $T\varphi:TM\to TN$.

As a bundle map, it can be seen that $\varphi\circ \pi_M =  \pi_N\circ T\varphi$, where $\pi_M:TM\to M$ and $\pi_N:TN\to N$ are canonical projections.
\end{definition}

\begin{definition}[Pullback of $k$-form]
Let $\varphi:M\to N$ be a smooth map between smooth manifolds $M$ and $N$, and let $\alpha \in \Omega^k(N)$ be a $k$-form on $N$. The pullback of $\alpha$ is a $k$-form on $M$, $\varphi^\ast\alpha\in\Omega^k(M)$, defined by $(\varphi^\ast\alpha)_x(v_1,\ldots,v_k) = \alpha_{\varphi(x)}(T_x\varphi(v_1),\ldots,T_x\varphi(v_k))$. In the case of a zero-form (or function) $f\in \Omega^0(N)$, the pullback reduces to precomposition $\varphi^\ast f = f\circ \varphi\in\Omega^0(M)$. 

The most important properties of the pullback is that it is compatible with the wedge product, $\varphi^\ast (\alpha\wedge \beta) = \varphi^\ast\alpha\wedge \varphi^\ast\beta$, and commutes with the exterior derivative, $\varphi^\ast (\extd\alpha)=\extd(\varphi^\ast \alpha)$.
\end{definition}

Now, let $M$ be an $m$-dimensional manifold and $g_t:M\to M$ a smooth family of diffeomorphisms (smooth mappings with smooth inverses) with parameter $t\in I\subseteq \R$ taking values in an open interval $I$. The sequence of mappings induces a curve $x(t)=g_t(x_0)\in M$ for each individual reference point $x_0\in M$. The reference point $x_0$ should not be interpreted as an initial condition but rather as a \emph{label} for the particle moving along the curve $x(t)$. (See Section 1 of Ref. \cite{Holm_1998} for a discussion of particle relabeling symmetry in fluid theories.) The time-derivative of such curve is a tangent vector at $x(t)$, i.e. $\dot{x}(t)=\partial_t g_t(x_0)=X_t(x(t))\in T_{x(t)}M$ where the time-dependent vector field $X_t:=\partial_t g_t\circ g_t^{-1}:M\to TM$ has been identified.

\begin{proposition}[Derivative of a pull back of a time-dependent function]\label{th:derivative_of_pullback_of_function}
Let $f_t:M\rightarrow\mathbb{R}$ be a time-dependent function. Let $X_t:M\to TM$ be the time-dependent vector field associated to diffeomorphism $g_t:M \rightarrow M$. The pullback of $f_t$ by $g_t$ satisfies
\begin{align*}
\partial_t (g_t^{\ast}f_t)=g_{t}^{\ast}(\partial_t f_t+\lie_{X_t}f_t).
\end{align*}
\end{proposition}
\begin{proof}
For $x_0\in M$ and its characteristic curve $x(t)=g_t(x_0)$, we have $g_t^\ast f_t(x_0)=f_t(x(t))$. 
Using the chain rule for differentiation, direct calculation gives
\begin{align*}
\partial_t(g_t^{\ast}f_t)(x_0)
&=\partial_t f_t (x(t)) +  {\extd f_t}_{x(t)}(\dot{x}(t))\\
&=\partial_t f_t(g_t(x_0)) + \lie_{X_t}f_t(g_t(x_0))\\
&=[g_t^\ast (\partial_t f_t + \lie_{X_t}f_t)](x_0).
\end{align*}
\end{proof}
\begin{proposition}[Derivative of a pullback of a time-dependent $k$-form]\label{th:derivative_of_a_pullback}
Given a time-dependent $k$-form $\alpha_t$ and a time-dependent vector field $X_t$ related to family of diffeomorphisms $g_t$, one has
\begin{align*}
\partial_t(g_t^{\ast}\alpha_t)=g_t^{\ast}(\partial_t\alpha_t+\lie_{X_t}\alpha_t)
\end{align*}
\end{proposition}
\begin{proof}
We first assume that it is true for the forms $\alpha_t$ and $\beta_t$. Then it is true for $\alpha_t\wedge\beta_t$ because
\begin{align*}
\partial_t(g_t^{\ast}(\alpha_t\wedge\beta_t))&=\partial_t(g_t^{\ast}\alpha_t)\wedge g_t^{\ast}\beta_t+g_t^{\ast}\alpha_t\wedge\partial_t(g_t^{\ast}\beta_t)\\
%&=g_t^{\ast}(\partial_t\alpha_t+\lie_{X_t}\alpha_t)\wedge g_t^{\ast}\beta_t+g_t^{\ast}\alpha_t\wedge g_t^{\ast}(\partial_t\beta_t+\lie_{X_t}\beta_t)\\
&=g_t^{\ast}((\partial_t\alpha_t+\lie_{X_t}\alpha_t)\wedge\beta_t+\alpha_t\wedge (\partial_t\beta_t+\lie_{X_t}\beta_t))\\
%&=g_t^{\ast}(\partial_t\alpha_t\wedge\beta_t+\alpha_t\wedge\partial\beta_t+\lie_{X_t}\alpha_t\wedge\beta_t+\alpha_t\wedge\lie_{X_t}\beta_t)\\
&=g_t^{\ast}(\partial_t(\alpha_t\wedge\beta_t)+\lie_{X_t}(\alpha_t\wedge\beta_t))
\end{align*}
and trivially also for $\extd\alpha_t$ and $\extd\beta_t$ as the exterior derivative commutes with the pullback and the Lie derivative. Proposition~\ref{th:derivative_of_pullback_of_function} shows that it is true for zero forms (functions). All other forms can be constructed from zero forms using combinations of $\extd$ and $\wedge$, so the result follows.
\end{proof}

\begin{corollary}
\label{th:derivative_pullbackbyinverse}
Let $\alpha$ be a (time-independent) $k$-form on $M$ and let $\alpha_t={g_t^{-1}}^\ast \alpha=: g_{t*}\alpha$ be the pushforward of $\alpha$ by $g_t$. Then,
\begin{align*}
    (\partial_t + \lie_{X_t})\alpha_t =0, 
\end{align*}
where $X_t=\partial_t g_t\circ g_t$ is the related vector field.
\end{corollary}
\begin{proof}
 Since $g_t^\ast  \alpha_t = \alpha$, $\partial_t g_t^\ast \alpha_t = \partial_t\alpha=0$, and the result follows from Proposition~\ref{th:derivative_of_a_pullback}.
\end{proof}

%\david{There is no such thing as a pushforward of a $k$-form. The corollary above replace the next theorem (to be deleted). }
%\josh{That's not right. What is right is that you cannot define the pushforward of a form along a map that is not a diffeomorphism. However pushforward along a diffeomorphism makes perfect sense. You can define it as the inverse of pullback. It's silly to always refer to inverse-of-pullback rather than pushforward though; which is quicker to write, in symbols and in english?}
%\begin{theorem}[Derivative of a pushforward of a time-independent $k$-form]\label{th:derivative_of_a_pushforward}
%Given a time-independent $k$-form $\alpha$ and a time-dependent vector field $X_t$ with a flow $g_t$, the time derivative of a pushforwarded form is
%\begin{align*}
%\partial_t g_{t\ast}\alpha=-\lie_{X_t}(g_{t\ast}\alpha)
%\end{align*}
%\end{theorem}
%\begin{proof}
%Choosing $\alpha_t=g_{t\ast}\alpha$ in Proposition~\ref{th:derivative_of_a_pullback}, one finds $\partial_t(g_t^{\ast}\alpha_t)=\partial_t(g_t^{\ast}g_{t\ast}\alpha)=\partial_t\alpha=0$, and the result follows using Proposition~\ref{th:derivative_of_a_pullback}.
%\end{proof}

\begin{proposition}[Derivative of a vector field of a two-parameter diffeomorphism]\label{th:derivative_of_vector_field}
Let $g_{t,s}:M\rightarrow M$ be a two-parameter family of diffeomorphisms with $(t,s)\in I\times J\subseteq \R^2$ (open), which generates the pair of two-parameter vector fields $X_{t,s}=\partial_tg_{t,s}\circ g_{t,s}^{-1}$ and $Y_{t,s}=\partial_s g_{t,s}\circ g_{t,s}^{-1}$. Then
\begin{align*}
\partial_sX_{t,s}-\partial_tY_{t,s}=[X_{t,s},Y_{t,s}]=-[Y_{t,s},X_{t,s}],
\end{align*}
where $[X,Y]=\lie_{X}Y$ is the Lie bracket of vector fields satisfying $\lie_{X}\lie_{Y}-\lie_{Y}\lie_{X}=\lie_{[X,Y]}$.
\end{proposition}
\begin{proof}
Let $f:M\rightarrow \mathbb{R}$ be a function and set $f_{t,s}=g_{t,s}^{\ast}f$. The partial derivatives with respect to $t$ and $s$ commute (Clairaut's theorem). Using Proposition~\ref{th:derivative_of_pullback_of_function} (first for time-independent function and then for a time-dependent function), we compute
\begin{align*}
0&=\partial_t\partial_sf_{t,s}-\partial_s\partial_tf_{t,s}\\
&=\partial_t(g_{t,s}^{\ast}(\lie_{Y_{t,s}}f))-\partial_s(g_{t,s}^{\ast}(\lie_{X_{t,s}}f))\\
&=g_{t,s}^{\ast}(\partial_t\lie_{Y_{t,s}}f+\lie_{X_{t,s}}\lie_{Y_{t,s}}f-\partial_s\lie_{X_{t,s}}f-\lie_{Y_{t,s}}\lie_{X_{t,s}}f)\\
&=g_{t,s}^{\ast}(\lie_{\partial_tY_{t,s}-\partial_sX_{t,s}+[X_{t,s},Y_{t,x}]}f)
\end{align*}
Because $f$ is arbitrary, the result follows.
\end{proof}

\section{Derivation of \eqref{eq:on-shell-variation-spatial}}\label{sec:on-shell-variation-spatial-details}
Assuming $\xi_t$, $A_t$, and $\phi_t$ to satisfy the Euler-Lagrange conditions \eqref{eq:Euler-Lagrange-xi}, \eqref{eq:Euler-Lagrange-A}, and \eqref{eq:Euler-Lagrange-phi}, a direct substitution to \eqref{eq:variation} with the replacements $Q\rightarrow U$ and $TQ\rightarrow TU$ provides
\begin{align}
    &\delta \mathfrak{S}_{U,t_1,t_2}[(\partial_t+\lie_{\xi_t})\widetilde{X},-\lie_{\widetilde{X}} f_{t},-\iprod_XB_t-\extd(\iprod_XA_t),\iprod_XE_t+\partial_t(\iprod_XA_t)]\nonumber
    \\
    &=\int\limits_{t_1}^{t_2}\int\limits_{TU}\left\{
    (\partial_t+\lie_{\xi_t})\left(\iprod_{\widetilde{X}}\vartheta f_t\right)
    -\lie_{\widetilde{X}}\left[(\iprod_{\xi_t}(\vartheta+e\pi^{\ast}A_t)-(K_t+e\pi^{\ast}\phi_t))f_t\right]\right\}dt \nonumber
    \\
    &-\int\limits_{t_1}^{t_2}\int\limits_{U}\left[\extd(\star D_t\iprod_XE_t-\iprod_XB_t\wedge\star H_t)-\partial_t(\iprod_XB_t\wedge
    \star D_t)\right]dt\nonumber
    \\
    &+\int\limits_{t_1}^{t_2}\int\limits_{U}\left\{\extd(\iprod_XA_t)\wedge(\star \partial_tD_t-\extd\star H_t)-(\iprod_X\partial_tA_t)\extd\star D_t \right\}dt 
    \nonumber 
    \\
    &+\int\limits_{t_1}^{t_2}\int\limits_{TU}(\partial_t+\lie_{\xi_t})\left(e\iprod_{\widetilde{X}}\pi^{\ast}A_tf_t\right)dt,
\end{align}
where only rearrangements have occurred. Since $f_t$ satisfies \eqref{eq:advection-f}, one  has
$(\partial_t+\lie_{\xi_t})\left(e\iprod_{\widetilde{X}}\pi^{\ast}A_tf_t\right)=ef_t\pi^{\ast}\iprod_X\partial_tA_t+ef_t\iprod_{\xi_t}\pi^{\ast}\extd(\iprod_XA_t)$. Using once more the equations \eqref{eq:Euler-Lagrange-A} and \eqref{eq:Euler-Lagrange-phi}, with the testing zero- and one-form being $\iprod_X\partial_tA_t$ and $\extd(\iprod_X A_t)$, respectively, one observes that the last two lines in the above variation cancel each other and \eqref{eq:on-shell-variation-spatial} follows.

\section{Derivation of \eqref{eq:on-shell-variation-temporal}}\label{sec:on-shell-variation-temporal-details}
Assuming $\xi_t$, $A_t$, and $\phi_t$ to satisfy the Euler-Lagrange conditions \eqref{eq:Euler-Lagrange-xi}, \eqref{eq:Euler-Lagrange-A}, and \eqref{eq:Euler-Lagrange-phi}, a direct substitution to \eqref{eq:variation} with the replacements $Q\rightarrow U$ and $TQ\rightarrow TU$ provides
\begin{align}
    &\delta \mathfrak{S}_{U,t_1,t_2}[(\partial_t+\lie_{\xi_t})(-\xi_t),-\lie_{-\xi_t} f_{t},E_t+\extd\phi_t,-\partial_t\phi_t]\nonumber
    \\
    &=-\int\limits_{t_1}^{t_2}\int\limits_{TU}
    \partial_t\left(\iprod_{\xi_t}(\vartheta+e\pi^{\ast}A_t)f_t-(K_t+e\pi^{\ast}\phi_t)f_t\right)dt\nonumber
    \\
    &-\int\limits_{t_1}^{t_2}\int\limits_{U}\left\{\extd (E_t\wedge\star H_t)+\partial_t(E_t\wedge\star D_t)\right\}dt-\int\limits_{t_1}^{t_2}\int\limits_{TU}(\partial_t+\lie_{\xi_t})(K_tf_t)dt\nonumber
    \\
    &-\int\limits_{t_1}^{t_2}\int\limits_{U}\left\{-\partial_t\phi_t\extd\star D_t+\extd\phi_t\wedge(\star \partial_tD_t-\extd\star H_t)\right\} dt-\int\limits_{t_1}^{t_2}\int\limits_{TU}(\partial_t+\lie_{\xi_t})(e\pi^{\ast}\phi_tf_t)dt,
\end{align}
where only rearrangements have occurred.  Since $f_t$ satisfies \eqref{eq:advection-f}, one  has
$(\partial_t+\lie_{\xi_t})\left(e\pi^{\ast}\phi_tf_t\right)=ef_t\pi^{\ast}\partial_t\phi_t+ef_t\iprod_{\xi_t}\pi^{\ast}\extd\phi_t$. Using once more the equations \eqref{eq:Euler-Lagrange-A} and \eqref{eq:Euler-Lagrange-phi}, with the testing zero- and one-form being $\partial_t\phi_t$ and $\extd\phi_t$, respectively, one observes that the terms on last line in the above variation cancel each other and \eqref{eq:on-shell-variation-temporal} follows.\\\\\\

\bibliography{references}  
\end{document}